\documentclass[conference,a4paper]{IEEEtran}%
\usepackage{latexsym,delarray,multicol}
\usepackage{epsfig}
\usepackage{amsfonts}
\usepackage{mathrsfs}
\usepackage{amsthm}
\usepackage{amsmath}
\usepackage{amssymb}
\usepackage{revsymb}
\usepackage{graphicx}%
\usepackage{url}
\usepackage{cite}

\newtheorem{corollary}{Corollary}

\newtheorem{proposition}{Proposition}
\newtheorem{lemma}{Lemma}
\begin{document}

\title{Distance verification for LDPC codes}
\date{\today}
\author{\IEEEauthorblockN{Ilya
Dumer\IEEEauthorrefmark{1}, Alexey A. Kovalev\IEEEauthorrefmark{2},
and Leonid P.~Pryadko\IEEEauthorrefmark{3}}
\IEEEauthorblockA{\IEEEauthorrefmark{1} ECE Department, University of California at Riverside, USA
(e-mail: dumer@ee.ucr.edu)}
\IEEEauthorblockA{\IEEEauthorrefmark{2} Department of Physics \& Astronomy,
University of Nebraska-Lincoln, USA (e-mail: alexey.kovalev@unl.edu)}
\IEEEauthorblockA{\IEEEauthorrefmark{3} Department of
Physics \& Astronomy, University of California at Riverside, USA
(e-mail: leonid@ucr.edu)}}
\maketitle
\vspace{-0.03in}

\begin{abstract}
The problem of \ finding code distance has been long studied for the generic
ensembles of linear codes and led to several algorithms that substantially
reduce exponential complexity of this task. However, no asymptotic
complexity bounds are known for distance verification in other ensembles of
linear codes. Our goal is to re-design the existing generic algorithms of
distance verification and derive their complexity for LDPC codes. We obtain
new complexity bounds with provable performance
expressed in terms of the erasure-correcting thresholds
 of long LDPC codes. These bounds exponentially reduce complexity
estimates known for linear codes.


\end{abstract}\vspace{-0.03in}

\section{Introduction}

This paper addresses the problem of finding code distances of LDPC codes with
provable complexity estimates. Note that finding code distance $d$ of a
generic code is an NP-hard problem. This is valid for both the exact setting
\cite{Vardy-1997} \ and the evaluation problem \cite{Dumer-2003},
\cite{Cheng-2009},\ where we only verify if $d$ belongs to some interval
$[\delta,c\delta]$ given some constant $c\in(1,2).$ To this end, all
algorithms of distance verification discussed in this paper have exponential
complexity $2^{Fn}$ in blocklength $n$ \ and our goal is to reduce the
complexity exponent $F$.

Below we address generic algorithms of \textit{distance verification} - known
for linear codes - and re-design these algorithms for LDPC codes. The main
problem is that such algorithms heavily rely on the  properties of the
randomly chosen generator (or parity-check) matrices. These properties have
not been proved (or do not hold) for the smaller ensembles of codes, such as
cyclic codes, LDPC codes, and others. Therefore, we will use a different
technique and derive complexity estimates for LDPC codes using a single
parameter, which is the erasure-correcting threshold of a specific code
ensemble. \ We then define this threshold via the average weight spectra of
LDPC codes. This technique is different from the generic approach. In
particular, we calculate the average complexity of distance verification and
then discard a vanishing fraction of \ LDPC codes that have atypically high
complexity. \ Our main result is the new complexity bounds for distance
verification for ensembles  of LDPC codes or other ensembles with a given
erasure-correcting threshold. These algorithms perform with
an arbitrarily high level of accuracy.

Here, however, we leave out some efficient algorithms that require more
specific estimates to perform distance verification with provable complexity.
\ In particular, we  do not address belief propagation (BP) algorithms, which
can  end at the stopping sets and therefore fail to furnish
distance verification with an arbitrarily high likelihood. Some other algorithms
also include impulse techniques \cite{Declercq-Fossorier-2008} that apply list decoding BP algorithms
to the randomly induced errors. Simulation results presented in  \cite{Declercq-Fossorier-2008} 
show that impulse techniques can also be effective in distance
verification albeit with a lesser fidelity.

\section{Background}

Let $C[n,k]$ be a linear binary code of length $n$ and dimension $k$. The
problem of verifying distance $d$ of a linear code (finding a minimum-weight
codeword) is related to the decoding problem: find an error of minimum weight
that gives the same syndrome as the received codeword. \ The number of
operations $N$ required for distance verification can usually be defined by
some positive exponent $F=\overline{\lim}$ ($\log_{2}N)/n$ as $n\rightarrow
\infty$. For example, for any code $C[n,k]$, inspection of all $2^{k}$
distinct codewords has (time) complexity exponent $F=R$, where $R=k/n$ is the
code rate. Given substantially large memory, one can instead consider the
syndrome table that stores the list of all $q^{r}$ syndromes and coset
leaders, where $r=n-k$. This setting gives (space) complexity $F=1-R$. \ 

To proceed with the more efficient algorithms, we also need to consider some
parameters of the shortened and punctured codes. Let $G$ and $H$ denote a
generator and parity check matrices of a code $C[n,k]$. \ Let $I$ be some
subset of $g\geq k$ positions and $J$ be the complementary subset of
$\eta=n-g,$ $\eta\leq r,$ positions. Consider the punctured code $\widehat
{C}_{I}=\{c_{I}:$ $c\in C\}$ generated by submatrix $G_{I}$ of size $k\times
g.$ The complementary shortened code $C_{J}=\{c_{J}:c_{I}=0\}$ has
parity-check matrix $H_{J}$ of size $r\times\eta.$ These matrices include at
most $k$ and $\eta$ linearly independent rows, respectively. Let $b\left(
G_{I}\right)  =k-\mathop\mathrm{rank}G_{I}$ and $b\left(  H_{J}\right)
=\eta-\mathop\mathrm{rank}H_{J}$ denote the co-ranks of these two matrices.
Throughout the paper, we use the following simple statement.

\begin{lemma}
\label{lm:weight}For any linear code $C[n,k],$ matrices $G_{I}$ and $H_{J}$
have equal co-ranks $b\left(  G_{I}\right)  =b\left(  H_{J}\right)  $ on the
complementary subsets $I$ and $J.$
\end{lemma}

\begin{proof}
Code $C_{J}$ has size $2^{\dim C_{J}}=2^{\eta-\mathrm{rank}H_{J}}=2^{b\left(
H_{J}\right)  }$ and contains all (shortened) code vectors $c$ with $c_{I}=0.$
On the other hand, \ for a given matrix $G_{I},$ there are $2^{b\left(
G_{I}\right)  }$ vectors $c$ with $c_{I}=0.$ Thus, $b\left(  G_{I}\right)
=b\left(  H_{J}\right)  .$
\end{proof}

Now consider the shortened codes $C_{J}$ of length $\mathcal{\eta}=\theta n$
taken over over different sets $J.$ Let codes $C_{J}$ have the average size
$N_{\theta}=2^{\dim C_{J}}.$ \ Then Markov's inequality gives another useful estimate.

\begin{corollary}
\label{cor:fraction}For any subset $J$ and any $t>0$, at most a fraction
$\frac{1}{t}$ of the shortened codes $C_{J}$ have size exceeding $tN_{\theta}$.
\end{corollary}

We will now consider two ensembles of regular LDPC codes.  Ensemble
$\mathbb{A}(\ell,m)$ is defined by the equiprobable $r\times n$ matrices $H$
that have all columns of weight $\ell$ and all rows of weight $m=\ell n/r.$
Below we take $m\geq\ell\geq3.$ This ensemble also includes a smaller LDPC
ensemble $\mathbb{B}(\ell,m)$ originally proposed by Gallager \cite{Gallager}.
For each code in $\mathbb{B}(\ell,m)$, its parity-check matrix $H$ consists of
$\ell$ horizontal blocks $H_{1},...,H_{\ell}$ of size $\frac{r}{\ell}\times
n.$ The first block $H_{1}$ consists of $m$ consecutive unit matrices of size
$\frac{r}{\ell}\times\frac{r}{\ell}$. Any other block $H_{i}$ is obtained by
some random permutation $\pi_{i}(n)$ of $n$ columns of $H_{1}.\ $ Ensembles
$\mathbb{A}(\ell,m)$ and $\mathbb{B}(\ell,m)$ have similar spectra and achieve
the best asymptotic distance for a given code rate $1-\ell/m$ \ among various
LDPC ensembles studied to date \cite{Litsyn-2002}.

Note that LDPC codes are defined by non-generic, sparse parity check matrices
$H_{J}$. Below we will relate the co-ranks $b_{J}=b\left(  H_{J}\right)  $ of
these matrices $H_{J}$ to the erasure-correcting thresholds of LDPC codes. In
doing so, we extensively use the average weight spectra derived for the
ensemble $\mathbb{B}(\ell,m)$ in \cite{Gallager} and for ensemble
$\mathbb{A}(\ell,m)$ in \cite{Litsyn-2002}. \ We note, however, \ that this
analysis can readily be extended to other ensembles with the known average
weight spectra.

Let $\alpha=\ell/m=1-R.$ For any parameter $\beta\in\lbrack0,1],$ consider the
equation
\begin{equation}
\frac{(1+t)^{m-1}+(1-t)^{m-1}}{(1+t)^{m}+(1-t)^{m}}=1-\beta\label{t1}%
\end{equation}
that has a single positive root $t$. Also, let $h(\beta)$ be the binary
entropy function. \ Below we extensively use the parameter
\begin{equation}
q\left(  \alpha,\beta\right)  =\alpha\log_{2}\frac{(1+t)^{m}+(1-t)^{m}%
}{2t^{\beta m}}-\alpha mh(\beta), \label{p1}%
\end{equation}
and also take $q\left(  \alpha,\beta\right)  =-\infty$  if
$m$ is odd and $\beta\geq1-\frac{1}{m}.$ Then Theorem 4 of \cite{Litsyn-2002}
shows that a given codeword of weight $\beta n$ belongs to some code in
ensemble $\mathbb{A}(\ell,m)$ with probability $P\left(  \alpha,\beta\right)  $ such
that
\begin{equation}
\underset{n\rightarrow\infty}{\lim}\textstyle\frac{1}{n}\log_{2}P\left(
\alpha,\beta\right)  =q\left(  \alpha,\beta\right)  \label{p2}%
\end{equation}

\begin{lemma}
\label{lm:ldpc}For any given subset $J$ of size $\theta n$, where $\theta
\leq1,$ codes $C_{J}(\ell,m)$ of the shortened ensemble $\mathbb{A}_{J}%
(\ell,m)$ have the average number $N_{\theta}$ of nonzero codewords such that
\begin{equation}
\underset{n\rightarrow\infty}{\lim}\textstyle\frac{1}{n}\log_{2}N_{\theta
}=f(\theta) \label{punct}%
\end{equation}
\begin{equation}
f(\theta):=\max_{0<\beta\leq1}\left\{  q\left(  \alpha,\beta\theta\right)
+{\theta h(\beta)}\right\}  \label{f}%
\end{equation}

\end{lemma}

\begin{proof}
For any set $J$ of size $\theta n,$ consider codewords $c$ of weight $\beta{\theta
n}$ that have support on $J.$ For any $\beta\in(0,1]$, codes in
$\mathbb{A}_{J}(\ell,m)$ contain the average number
$N_{\theta}\left(  \beta\right)  =P\left(  \alpha,\beta\theta\right)
{\binom{\theta n}{\beta\theta n}}$
of such codewords $c$.  Then
\begin{equation}
\max_{0<\beta\leq1}\frac{\log
_{2}N_{\theta}\left(  \beta\right)  }{n}\,{\sim}\,\max_{0<\beta\leq1}\left\{
q\left(  \alpha,\beta\theta\right)  +{\theta h(\beta)}\right\}  \smallskip
\smallskip\label{p3}%
\end{equation}
which gives asymptotic equalities (\ref{punct}) and (\ref{f}).
\end{proof}

\section{Distance verification for LDPC codes\label{sec:LDPC}}\vspace{-0.05in}

\subsection{Two main parameters for complexity estimates.
\ \label{sec:SW-LDPC}}

Two essential differences separate LDPC ensembles from random codes in regards
to complexity estimates. These differences are closely related to two
parameters, $\delta_{\ast}$ and $\theta_{\ast},$ which are the roots of the
equations%
\begin{equation}%
\begin{tabular}
[c]{l}%
$\delta_{\ast}:h(\delta_{\ast})+q(\alpha, \delta_{\ast})=0\smallskip
\smallskip$\\
$\theta_{\ast}:f(\theta_{\ast})=0.\,$%
\end{tabular}
\ \ \ \ \label{dist1}%
\end{equation}
Note that $\delta_{\ast}$ is the average relative code distance \ in the
ensemble $\mathbb{A}(\ell,m).$ \ Indeed, for ${\theta=1,}$ equality (\ref{p3})
shows that the average number of codewords $N_{\theta}(\beta)$ of length $n$
and weight $\beta n$ has asymptotic order
\begin{equation}
\textstyle\frac{1}{n}\log_{2}N(\beta)\,{\sim}\,h(\beta)+q\left(  \alpha
,\beta\right)  \label{spec}%
\end{equation}
For any code rate $R=1-\ell/m$, $\delta_{\ast}$ \ falls below the GV distance
$h^{-1}(1-R)$ of random codes (see \cite{Gallager} and \cite{Litsyn-2002} for
more details). \ For example, $\delta_{\ast}\sim0.02$ for the $\mathbb{A}%
(3,6)$\ ensemble of rate $R=1/2$, whereas $h^{-1}(0.5)\sim0.11.$ The smaller
distances $\delta_{\ast}$ will reduce the complexity of distance verification.

Parameter $\theta_{\ast}$ also plays a significant role in distance
verification. Namely, consider a code ensemble $\mathbb{C}$ of \ growing
length $n\rightarrow\infty.$ Let $N_{\theta}$ be the number of nonzero
codewords in the shortened codes $C_{J}$ averaged over all codes
$C\in\mathbb{C}$ and all subsets $J$ of size $\theta n.$ Then we use the
following statement. \ 

\begin{lemma}
\label{lm:erasure}Let the ensemble $\mathbb{C}$ have a vanishing average
number $N_{\theta}\rightarrow0$ of nonzero codewords in the shortened codes
$C_{J}$ of length $\theta n.$ Then most codes $C\in\mathbb{C}$ \ correct most
erasure subsets $J,$ with the exception of a vanishing fraction of \ codes $C$
and subsets $J.$
\end{lemma}

\begin{proof}
A code $C\in\mathbb{C}$ fails to correct some erasure set $J$ \ of weight
$\theta n$ iff code $C_{J}$ has $N_{J}(C)\geq1$ nonzero codewords. \ Let
$M_{\theta}$ be the average fraction of such codes $C_{J}$ taken over all
codes $C$ and all subsets $J.$ Note that $M_{\theta}\leq N_{\theta}$.  Per Markov's
inequality,  no more than a fraction $\sqrt{M_{\theta}}$ of codes $C$ may leave a
fraction $\sqrt{M_{\theta}}$ of sets $J$ \ uncorrected. \smallskip
\end{proof}

More generally, we say that an ensemble of codes $\mathbb{C}$ has the
erasure-correcting threshold $\theta_{\ast}$ if $N_{\theta}\rightarrow0$ for
any $\theta<\theta_{\ast}$ and $N_{\theta}\geq1$ for any $\theta>\theta_{\ast
}$ on the sets $J$ of size $\theta n$. \ Here ensembles $\mathbb{A}(\ell,m)$
and $\mathbb{B}(\ell,m)$ satisfy Lemma \ref{lm:erasure} for any $\theta
<\theta_{\ast}$ of (\ref{dist1}). Thus, $\theta_{\ast}$\ serves as a
lower bound on the erasure-correcting capacity of LDPC codes under ML
decoding. Alternatively, one can use other thresholds, such as the threshold
for message-passing algorithms. Note also that ensembles $\mathbb{A}(\ell,m)$
and $\mathbb{B}(\ell,m)$ are permutation-invariant and therefore yield
the same fraction of uncorrected codes $C$ for each erasure subset $J.$ Then
for any $\varepsilon>0$, the bound $N_{J}(C)\leq2^{\varepsilon n}$ holds on
\textit{all} subsets $J$ (except for a fraction of $2^{-\varepsilon n}$ of codes
$C)$ as long as $N_{\theta}\leq1.$

For LDPC codes, $\theta_{\ast}<\alpha,$ where $\alpha=1-R$ is the
erasure-correcting threshold for random linear codes. For example,
$\theta_{\ast}=0.483$ for the ensemble $\mathbb{A}(3,6)$ of LDPC codes. See
also papers \cite{Luby-2001,Rich-Urb-2001,Rich-2001,Pishro-2004}, where
parameter $\theta_{\ast}$\ is discussed in a greater detail for both ML
decoding and message-passing decoder.

 The reduced erasure-correcting threshold
$\theta_{\ast}$ will increase complexity estimates for LDPC codes. \ In the
sequel, we will show that the first factor (the smaller distance $\delta_{\ast}$) outweighs the second factor (the smaller threshold $\theta_{\ast}$) and
reduces complexity of distance verification for LDPC codes.\vspace{-0.05in}

\subsection{Sliding window (SW) technique for LDPC codes \label{sec:sliding}}

This technique of \cite{Evseev-1983} decodes
generic linear codes $C[n,k,d]$ generated by the randomly chosen \ $(Rn\times
n)$ matrices $G$. Note that most such codes have full dimension
$k=Rn$ and meet the asymptotic GV bound $d/n\rightarrow h^{-1}(1-R)$. It is
 shown in \cite{Evseev-1983} that nearly full decoding (that has error probability similar to that of ML decoding) can be performed
for most codes $C[n,k,d]$ with complexity of order $2^{nR(1-R)}.$ Below we
modify this algorithm for other ensembles of codes, such as $\mathbb{A}%
(\ell,m)$ or $\mathbb{B}(\ell,m)$.

\begin{proposition}
\label{prop:SW-1}Consider any ensemble of codes $\mathbb{C}$ with an average
relative distance {$\delta_{\ast}$} and an erasure-correcting bound
$\theta_{\ast}.$ For most codes $C\in\mathbb{C}$, SW technique performs
distance verification with complexity of exponential order $2^{Fn}$ or less,
where
\begin{equation}
F={(1-\theta}_{\ast}{)h(\delta_{\ast})} \label{sw-ldpc1}%
\end{equation}

\end{proposition}

\begin{proof}
Consider a sliding window $I(i,s)$, which is the set of $s$ cyclically
consecutive positions for some $i=0,\ldots,n-1$. \ We choose $s=(1-\theta
_{\ast}+\varepsilon)n,$ \ where $\varepsilon>0$ is a parameter such that
$\varepsilon\rightarrow0$ as $n\rightarrow\infty.$ A window $I(i,s)$ can
change its Hamming weight only by one when it moves from position $i$ to
$i+1$; thus any codeword $c$ of weight $d=\delta_{\ast}n$ has at least one
window $I(i,s)$ with the average Hamming weight $v=\left\lfloor \delta_{\ast
}s\right\rfloor $. \ \ For each window $I,$ we inspect all $L={\binom{s}{v}}$
vectors $c_{I}$ of weight $v.$ Here%
\[
\textstyle\frac{1}{n}\log_{2}L{\,{\sim}\,(1-\theta}_{\ast}+\varepsilon
{)h(\delta_{\ast})}%
\]
We then encode each vector $c_{I}$ performing erasure correction
on the complementary sets $J=\overline{I}$ of size $(\theta_{\ast}%
-\varepsilon)n.$  Thus, a typical vector $c_{I}$ generates the average number
$N_{\theta}$ of nonzero codewords $c_{J}$. Given $L$ vectors $c_{I}$ and $n$ sets $I=I(i,s),$ we obtain the
average encoding complexity of $n^{3}N_{\theta}L$. Here we take the average over different codes $C\in\mathbb{C}$.
Thus, at most a vanishing fraction $n^{-1}$ of such codes \ have complexity
above $n^{4}N_{\theta}L$ \ for all $n$ subsets $I.$ This gives (\ref{sw-ldpc1}%
) as $\varepsilon\rightarrow0.$\smallskip\ 
\end{proof}\vspace{-0.05in}

\subsection{Matching Bipartition (MB) technique for LDPC
codes\label{sec:bipartition}}

Below we briefly discuss MB-technique of \cite{Dumer-1986, Dumer-1989}. It
works for any linear code and yields the lowest asymptotic complexity  for
very high code rates $R\rightarrow1.$

\begin{proposition}
\label{prop:MB-2}MB technique performs distance verification for a linear code
of distance {$\delta_{\ast}n$} with complexity of exponential order $2^{Fn}$,
where
\begin{equation}
F={h(\delta_{\ast})}/2 \label{mb-2}%
\end{equation}

\end{proposition}

\begin{proof}
To find an (unknown) vector $e$ of weight $d={\delta_{\ast}n},$ we use the
\textquotedblleft left" window $I_{\ell}$ of length $s_{\ell}=\left\lfloor
n/2\right\rfloor $ starting in any position $i$ and the complementary
\textquotedblleft right" window $I_{r}$ of length $s_{r}=\left\lceil
n/2\right\rceil $ . At least one choice of $i$ gives the average weights
$v_{\ell}=\left\lfloor d/2\right\rfloor $ and $v_{r}=\left\lceil
d/2\right\rceil $ for truncated vectors $e_{\ell}$ and $e_{r}$ in windows
$I_{\ell}$ and $I_{r}$. The number $L$ of vectors $e_{\ell}$ and $e_{r}$ has
the order of
\[
\textstyle\frac{1}{n}\log_{2}L{\,{\sim}}\frac{1}{n}\log_{2}{\binom{s_{r}%
}{v_{r}}\,{\sim}}\,h({\delta_{\ast}})/2{\,}%
\]
We calculate the syndromes of all vectors $e_{\ell}$ and $e_{r}$ and try to
match two vectors with equal syndromes. This matching is performed by sorting
the elements of the combined set with complexity of order $Ln\log_{2}L$, which
gives exponent (\ref{mb-2}).\smallskip\ 
\end{proof}

Exponents (\ref{sw-ldpc1}) and (\ref{mb-2}) give the combined estimate
\begin{equation}
F=\min\{{(1-\theta}_{\ast}{)h(\delta_{\ast})},{h(\delta_{\ast})}/2\}
\label{comb}%
\end{equation}
Here parameters {$\delta_{\ast}$} and ${\theta}_{\ast}$ are defined for LDPC
codes in (\ref{dist1}).\smallskip\vspace{-0.05in}

\subsection{Covering set (CS) technique for LDPC\ codes\ \label{sec:atyp-LDPC}%
}

This probabilistic technique was proposed in \cite{Prange-1962} and has become
a benchmark in cryptography since the classical paper \cite{McEliece-1978}.
\ It lowers complexity estimate \ (\ref{comb}) for all but very high code
rates\ $R\rightarrow1.$ CS technique has also been studied for distance
verification of specific code families (see \cite{Lee-1988} and
\cite{Leon-1988}); however, provable results
\cite{Kruk-1989,Coffey-Goodman-1990} are only known for generic random codes. \ 

Below we choose any LDPC ensemble and describe CS technique in the following proposition.

\begin{proposition}
\label{prop:CS-1}Consider any code ensemble $\mathbb{C}$ with an average
relative distance {$\delta_{\ast}$} and an erasure-correcting bound
$\theta_{\ast}.$ For most codes $C\in\mathbb{C}$, CS technique performs
distance verification or corrects up to $\delta_{\ast}n$ errors with
complexity of order $2^{Fn}$ or less, where
\begin{equation}
F=h(\delta_{\ast})-{\theta}_{\ast}h(\delta_{\ast}/{\theta}_{\ast})
\label{cs-exp}%
\end{equation}

\end{proposition}

\begin{proof}
Let $e$ be some unknown codeword of weight $d$ in a given code $C\in
\mathbb{C}.$ Alternatively, we can consider an error vector $e$ of weight $d.$
To find $e,$ we repeatedly try to cover all $d$ nonzero positions of $e$ with
some subsets $J=\{i_{1},...,i_{s}\}$ of \ $s={\theta n}$ positions, where
${\theta=\theta_{\ast}-\varepsilon}$ and ${\varepsilon\rightarrow0}$ as
$n\rightarrow\infty.$ To cover every possible $d$-set, we need no less than
\[
T(n,s,d)={\binom{n}{d}}/{\binom{s}{d}}%
\]
sets $J.$ Below we randomly choose a larger number of
\begin{equation}
T=T(n,s,d)n\ln n \label{eq:112}%
\end{equation}
sets $J.$\ Following Theorem 13.4 of \cite{Erdos-book} it is easy to see that
$T$ trials fail to yield such an $(n,s,d)$-covering with a probability less
than $e^{-n\ln n}.$

Recall that $N_{\theta
}\rightarrow0$ for  the shortened codes $C_{J}$. \ Let $C_{J}(b)$ be a code
that contains $2^{b}-1$ nonzero codewords for some $b=0,...,{\theta n.}$ Also,
let $\alpha_{\theta}(b)$ be the fraction of codes $C_{J}(b)$ in the ensemble
$\mathbb{C}_{J}.$ Then
\begin{equation}
N_{\theta}=\sum_{b=0}^{{\theta n}}\left(  2^{b}-1\right)  \alpha_{\theta}(b)
\label{ave1}%
\end{equation}
A parity-check matrix $H_{J}$ of any code $C_{J}(b)$ \ has rank $s-b$ and size
$r\times s,$ where $r=n-k$ is the number of parity checks$.$ By Gaussian
elimination, matrix $H_{J}$ can be modified into a new $r\times s$ matrix
$\mathcal{H}_{J}$ that includes $b$ \ zero rows$.$ We will also place $s-b$ unit columns 
$u_{i}=(0...01_{i}0...0)$  in the first
positions $i\in\lbrack1,s-b]$ of $\mathcal{H}_{J},$
and $b$ other columns $g_{j}$ in the last positions
$j\in\lbrack s-b+1,s].$ Let $v=\mathcal{H}e^{T}$ denote the syndrome of vector
$e$ (possibly modified by the Gaussian elimination procedure).

$A.$ First, consider  general  error correction given a
syndrome $v\neq 0.$ If $b=0$ in a given trial $J,$ then matrix $\mathcal{H}_{J}$ has
full rank and we obtain vector $e$ of weight $wt(v).$ If $b>0,$ we assume that
$v$ contains only zero symbols in the last $b$ positions. Then CS algorithm
inspects all $2^{b}$ linear combinations (LC) of the last columns $g_{j}.$ Let
$LC(p)$ denote some LC that includes $p$ columns. If $LC(p)+v$ \ has weight
$w,$ we obtain vector $e$ of weight $w+p$ by adding $w$ unit columns $u_{i}.$

The overall decoding algorithm successively tries to find a vector $e$ of
weight $d=1,2....$ For any given $d,$ it runs over all subsets $J$ and ends
once we find a vector $e$ of weight $w+p=d.$ For any given code $C_{J}(b),$
this procedure includes one Gaussian elimination and up to $b$ vector
additions, which gives complexity $\mathcal{D}_{\theta}(b)\leq n^{3}%
+rb2^{b}\leq n^{3}2^{b}.$ \ \ For a given set $J,$ different codes $C_{J}(b)$
yield the average complexity%
\begin{equation}
\mathcal{D}_{\theta}(J)\leq\sum_{b=0}^{{\theta n}}n^{3}2^{b}\alpha_{\theta
}(b)=n^{3}N_{\theta}+n^{3} \label{ave2}%
\end{equation}
Thus, CS algorithm has the total average complexity $\mathcal{D}_{ave}\sim
n^{3}T$ for all $T$ sets $J.$ Then at most a vanishing fraction $1/n$ of codes
$C$ have complexity $\mathcal{D}\geq n^{4}T,$ which gives the exponent
 (\ref{cs-exp}) for the remaining codes in ensemble $\mathbb{C}$ as $n\to\infty$.

$B.$ Vector $e$ forms a codeword with syndrome $v=0.$ Then any code $C_{J}(0)$
has no nonzero codewords, and CS algorithm skips the above case $b=0$. Also,
we consider only $2^{b}-1$ nonzero combinations LC$(p)$ for the last $b$
columns in any $C_{J}(b)$. \ Thus, we replace (\ref{ave2}) with a similar
inequality%
\begin{equation}
\mathcal{D}_{\theta}(J)\leq\sum_{b=1}^{{\theta n}}n^{3}\left(  2^{b}-1\right)
\alpha_{\theta}(b)\leq n^{3}N_{\theta}+n^{3} \label{ave3}%
\end{equation}
that satisfies complexity bound (\ref{cs-exp}).\smallskip\smallskip
\end{proof}

\textit{Remark. \ }The existing CS algorithms  employ
some stringent properties of \ random ensembles of linear codes. For example,
the algorithm of \cite{Kruk-1989} uses the fact that most random binary
$r\times n$ matrices $H$,  except an exponentially small
fraction ${\binom{n}{r}}^{1-c}$ for $c>1,$  have \ all $r\times r$ submatrices $H_{J}$ \ with
nearly-full rank $r-b,$ where
\begin{equation}
0\leq b\leq b_{\max}=\textstyle\sqrt{c\log_{2}{\binom{n}{r}}} \label{11}%
\end{equation}
 Thus,  all shortened codes $C_{J}$ have limited  size $2^{b}$ for most linear codes $C.$
For LDPC codes, we use a slightly weaker condition. Our technique  discards  codes $C_{J}$ of
large size $2^{b}$ that form an exponentially small fraction of all codes
$C_{J}.$

Fig. \ref{fig:ldpc} summarizes complexity estimates for LDPC codes. For
comparison, we also plot two generic exponents valid for most linear codes.
Note that these codes meet the\ GV bound and have parameters $h(\delta_{\ast
})={\theta}_{\ast}=1-R.$ \ Then the combination (\ref{comb}) of SW and MB
algorithms\ gives exponent $F=\min\{R(1-R),(1-R)/2\}$, whereas exponent
(\ref{cs-exp}) of \ CS algorithm reads as $F=(1-R)\bigl[1-h\left(
\delta/(1-R)\right)  \bigr].$ For LDPC codes, we similarly consider the
exponents (\ref{comb}) and (\ref{cs-exp}). Here we consider ensembles
$\mathbb{A}(\ell,m)$ or $\mathbb{B}(\ell,m)$ for various LDPC $(\ell,m)$ codes
with code rates ranging from $0.125$ to $0.8.$ With the exception of low-rate
codes, all LDPC codes of Fig. \ref{fig:ldpc} achieve a substantial reduction
in complexity exponent for distance verification compared to the generic
linear codes.

\begin{figure}[ptbh]
\centering
\includegraphics[width=1.\columnwidth]{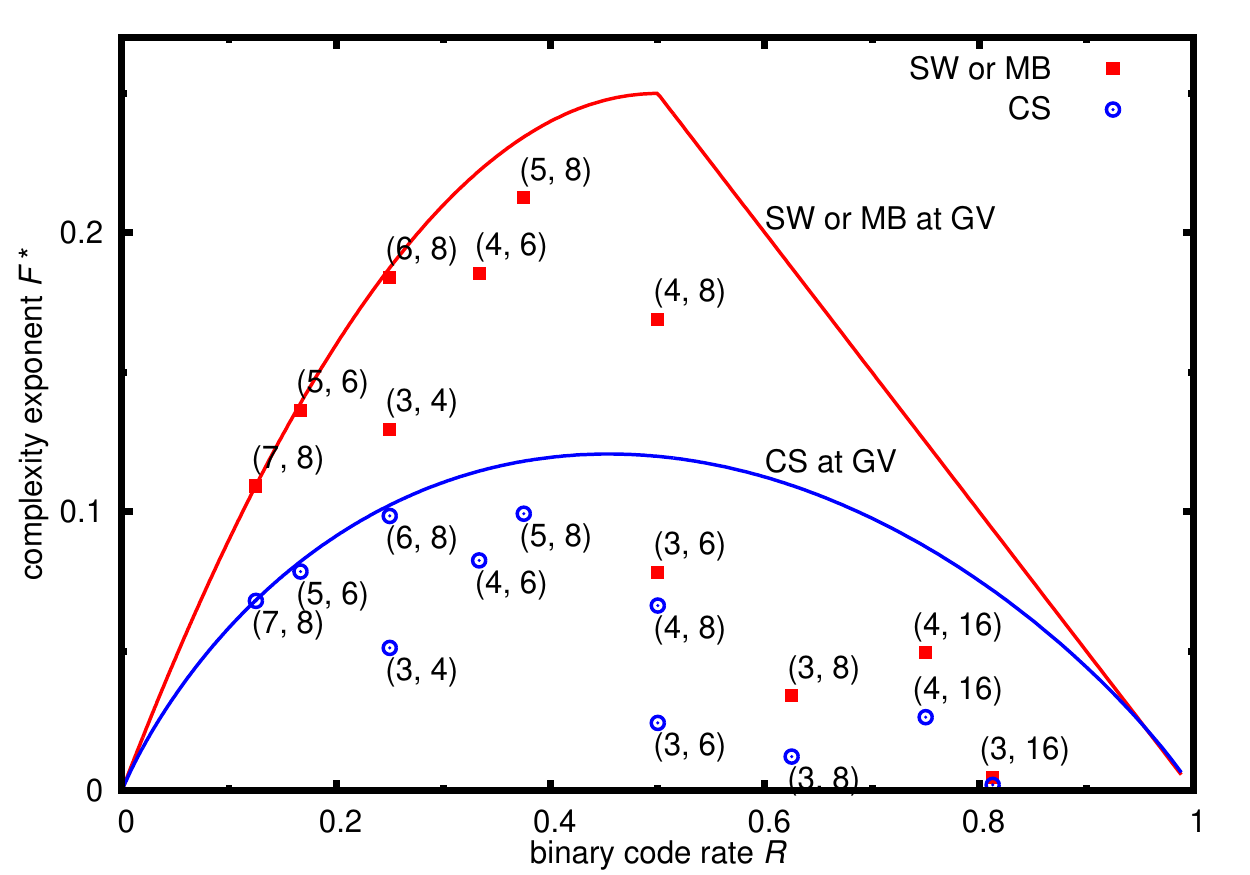} \caption{Complexity exponents
for the binary codes meeting the GV bound and for some ($\ell,m)$-regular LDPC
codes as indicated. Abbreviation \textquotedblleft SW or MB\textquotedblright%
\ stands for the Sliding Window or Matching  Bipartition techniques (marked with filled boxes), and "CS"
stands for the Covering-Set technique (marked with  empty circles).}

\label{fig:ldpc}%
\end{figure}

\section{Further extensions}

In this paper, we study provable algorithms of distance verification for LDPC
codes and derive complexity estimates using only the relative distance
$\delta_{\ast}$ and the erasure-correcting threshold ${\theta}_{\ast}$
averaged over a given ensemble of codes. For LDPC codes,  these algorithms exponentially reduce
generic complexity estimates known for 
random linear codes. More generally, this approach can be used for any
ensemble of codes with a given erasure-correcting threshold. \ 

One particular extension is any \ ensemble of irregular LDPC codes with the
known parameters $\delta_{\ast}$ and ${\theta}_{\ast}.$ Note that parameter
${\theta}_{\ast}$ has been studied for both ML decoding and message-passing
decoding of irregular codes
\cite{Luby-2001,Rich-Urb-2001,Pishro-2004}. For ML decoding, this
parameter can also be derived using the weight spectra obtained for irregular
codes in papers \cite{Di-2001,Litsyn-2003}.

Another direction is to design more advanced algorithms of distance
verification  for LDPC codes. Most of such algorithms
known to date for linear $[n,k]$ codes combine Matching Bipartition (MB)
techniques with the Covering Set (CS) algorithms.  In particular, the algorithm
of \cite{Stern-1989} first applies CS technique seeking some slightly
corrupted information set of $k$ bits. It also tries to select some small subset of
$\Delta$ parity bits,  every time  assuming that these bits are error-free. 
Then  MB technique is applied  to correct  information bits in the $[k+\Delta,k]$-code with  $\Delta$ correct parity bits. 
\ This algorithm reduces the maximum
complexity exponent $\max_{R}F(R)$ of CS technique from 0.1208 to 0.1167. A
slightly more efficient algorithm of \cite{Dumer-1991} (see also
\cite{Barg-1998}) reduces this exponent to 0.1163 using a lightly corrupted
block of length greater than $k$. Later, this algorithm has been
re-established for cryptographic setting in \cite{Sendrier-2009,
Bernstein-2011} with many applications related to the McEliece cryptosystem.
More recently, the maximum complexity exponent $F(R)$ has been further reduced to 0.1019
using some robust MB techniques that allow randomly overlapping partitions
\cite{Becker-2012}. An important observation is that both MB and CS techniques
can be applied to LDPC codes; therefore our conjecture is that provable
complexity bounds for distance verification also  carry over to the above
techniques. These more advanced algorithms can again slightly reduce the
exponent of CS complexity for LDPC codes; however, their detailed description
is beyond the scope of this paper.

Finally, one more approach is to combine LDPC-specific message-passing
algorithms with the subsequent erasure correction. Such an approach has been
used in \cite{Dumer-Kovalev-Pryadko-2014} for quantum LDPC codes that require
stringent self-orthogonality conditions. \ The corresponding complexity
exponent closely approaches exponent $F(R)$ for self-orthogonal LDPC codes
that have high code rate and low distance. For all other instances, this
approach requires substantial improvements as complexity exponents exceed the
exponent $F(R)$ obtained in the current paper.\smallskip

{\it Acknowledgment. } The work of L.P. Pryadko was supported in part by ARO Grant W911NF-14-1-0272
and NSF Grant PHY-1416578. The work of A.A. Kovalev was supported in part by
NSF Grant PHY-1415600. \smallskip





\bibliographystyle{IEEEtran}
\bibliography{IEEEabrv,ldpc1,lpp,qc_all,more_qc,MyBIB}
$\bigskip\bigskip$

\end{document}